\documentclass[letterpaper, 10 pt, conference]{ieeeconf}
\IEEEoverridecommandlockouts
\usepackage{graphics}
\usepackage{epsfig}
\usepackage{amsmath}
\usepackage{amssymb}
\usepackage{xcolor}
\let\labelindent\relax
\usepackage{enumitem}
\usepackage{multirow}
\usepackage{url}

\usepackage{amsthm}
\theoremstyle{definition}

\newtheorem{theorem}{\normalfont\bfseries Theorem}

\newtheorem{definition}{\normalfont\bfseries Definition}

\newtheorem{remark}{\normalfont\bfseries Remark}
\newtheorem{example}{\normalfont\bfseries Example}

\newcommand{\R}{\mathbb{R}}
\newcommand{\X}{\mathbb{R}^n}
\newcommand{\U}{\mathbb{R}^m}
\newcommand{\C}{\mathcal{C}}
\newcommand{\K}{\mathcal{K}}
\newcommand{\Ke}{\K^{\rm e}}

\newcommand{\derp}[2]{\frac{\partial #1}{\partial #2}}
\renewcommand{\H}{\mathcal{H}}
\renewcommand{\P}{\mathcal{P}}

\newcommand{\bzero}{\mathbf{0}}

\newcommand{\bx}{\mathbf{x}}
\newcommand{\bu}{\mathbf{u}}
\newcommand{\bff}{\mathbf{f}}
\newcommand{\bg}{\mathbf{g}}
\newcommand{\bk}{\mathbf{k}}

\newcommand{\bn}{\mathbf{n}}
\newcommand{\bp}{\mathbf{p}}
\newcommand{\bw}{\mathbf{w}}

\newcommand{\bq}{\mathbf{q}}
\newcommand{\bv}{\mathbf{v}}
\newcommand{\bD}{\mathbf{D}}
\newcommand{\bC}{\mathbf{C}}
\newcommand{\bG}{\mathbf{G}}
\newcommand{\bB}{\mathbf{B}}

\title{\LARGE \bf
Navigating Polytopes with Safety: A Control Barrier Function Approach
}

\author{Tamas G. Molnar%
\thanks{*The material contained in this document is based upon work supported by a National Aeronautics and Space Administration (NASA) grant or cooperative agreement. Any opinions, findings, and conclusions or recommendations expressed in this material are those of the author(s) and do not necessarily reflect the views of NASA. This work was supported through a NASA grant awarded to the Kansas NASA EPSCoR Program.}%
\thanks{Department of Mechanical Engineering, Wichita State University, Wichita, KS 67260, USA,
{\tt\small tamas.molnar@wichita.edu}.}%
}

\begin{document}

\maketitle
\thispagestyle{empty}
\pagestyle{empty}

%%%%%%%%%%%%%%%%%%%%%%%%%%%%%%%%%%%%%%%%%%%%%%%%%%%%%%%%%%%%%%%%%%%%%%%%%%%%%%%%
\begin{abstract}
Collision-free motion is a fundamental requirement for many autonomous systems.
This paper develops a safety-critical control approach for the collision-free navigation of polytope-shaped agents in polytope-shaped environments.
A systematic method is proposed to generate control barrier function candidates in closed form that lead to controllers with formal safety guarantees.
The proposed approach is demonstrated through simulation, with obstacle avoidance examples in 2D and 3D, including dynamically changing environments.
\end{abstract}
% \vspace{-1mm}

%%%%%%%%%%%%%%%%%%%%%%%%%%%%%%%%%%%%%%%%%%%%%%%%%%%%%%%%%%%%%%%%%%%%%%%%%%%%%%%%
\section{INTRODUCTION}
\label{sec:intro}

Collision-free navigation of controlled agents, such as robotic systems or intelligent vehicles, is crucial for safe autonomous behavior.
While several motion planning approaches exist to tackle this problem, recent research also focuses on reactive safety-critical control, primarily to address dynamically changing or unknown environments where planning is difficult.
In this domain, control barrier function (CBF)~\cite{AmesXuGriTab2017, garg2024advances} has become a dominant tool for safe control design.
In this paper, we focus on CBFs for safe navigation.

A variety of CBF-based approaches exists for safe navigation of various geometries.
For example, \cite{hamatani2020collision} controlled planar robot arms to avoid points, while
\cite{almubarak2022safety}
% planned and executed safe optimal trajectories for
navigated point agents among various obstacle shapes (ellipse, cardioid, diamond, square, and sphere) via discrete barrier states and differential dynamic programming.
Moreover,~\cite{mestres2024distributed} established CBFs for multi-agent systems based on circular agent and obstacle geometries, whereas~\cite{landi2019safetybarrier} controlled manipulators to interact with humans whose geometry was described using capsules (cylinders with two hemispheres).
For general agent and obstacle geometries,~\cite{singletary2022manipulation} used the signed distance function as CBF.
In unknown environments, several works discussed safe navigation with CBFs based on perception.
These include
signed distance approximations using support vector machines~\cite{srinivasan2020synthesis} and neural networks~\cite{long2021learning};
sensor-based robot navigation in human crowds using model predictive control with CBFs~\cite{vulcano2022};
vision-based CBFs for vehicles~\cite{abdi2023safecontrol} and
LiDAR-based CBFs for mobile robots~\cite{keyumarsi2024lidar};
safe vision-based control with CBFs over point clouds~\cite{desa2024pointcloud, dai2024sailing} and Gaussian splatting maps~\cite{chen2024safer};
and neural navigation CBFs~\cite{harms2024neural}.

In this paper, we focus on scenarios where the geometries of both the navigating agent and its environment are described by polytopes.
Related works on CBFs for polytope geometries include~\cite{long2024safestabilizing} where polygon robots were navigated in dynamic elliptic environments by computing the polygon-ellipse distance with an analytic approach.
Moreover,~\cite{thirugnanam2022safety} addressed obstacle avoidance between polytopes via
% planning and control via
discrete-time CBF constraints in model predictive control, and
\cite{thirugnanam2022duality} extended this work to continuous time using a duality-based safety-critical optimal control approach with nonsmooth CBFs.
Finally,~\cite{tayal2024polygonal} established polygonal cone CBFs for navigation in cluttered environments by drawing collision cones based on the vertices of polygon obstacles.
% These works mostly focus on 2D scenarios or in 3D they consider agents as points with a buffer for their size.

A challenge in safe navigation is that infinitely many points (the points of the agent) must be kept inside a safe set (the collision-free space of the environment).
Meanwhile, classical CBF formulations keep a single point inside a set.
If the agent and environment geometries are (approximated by) polytopes, however, the navigation problem reduces to a combination of finitely many safety constraints which can be encoded into CBFs.
Motivated by this, here we propose an algorithmic method to construct a smooth CBF candidate for polytope agents navigating in known polytope environments.
The uniqueness of this method is that the resulting CBF candidate is given in a simple closed form, which allows the use of analytic expressions for safe controllers even in complex real-world settings, facilitating real-time implementation and offering simplicity over existing optimization-based methods.
As demonstrated by simulation examples, the polytopes may not be convex,
% they may have arbitrary dimension, meaning that
and the method is applicable to both 2D and 3D navigation, including dynamically changing environments.
This method opens the way to address arbitrary agent and environment geometries via polytope approximations.

\section{BACKGROUND}
\label{sec:CBF}

Control barrier functions provide a constructive method for designing safe controllers for systems of the form:
\begin{equation} \label{eq:system}
    \dot{\bx} = \bff(\bx) + \bg(\bx) \bu,
\end{equation}
where ${\bx \in \X}$ is the state, ${\bu \in \U}$ is the input, and the locally Lipschitz functions ${\bff : \X \to \R^n}$ and ${\bg: \X \to \R^{n \times m}}$ describe the dynamics.
Using a locally Lipschitz state feedback controller ${\bk : \X \to \U}$, ${\bu = \bk(\bx)}$, we intend to ensure that the solution of the closed-loop system:
\begin{equation} \label{eq:closedloop}
    \dot{\bx} = \bff(\bx) + \bg(\bx) \bk(\bx)
\end{equation}
evolves inside a safe set ${\C \subset \X}$.
We call the system~\eqref{eq:closedloop} safe w.r.t.~$\C$ if for any initial condition ${\bx(0) = \bx_{0} \in \C}$ the unique solution of~\eqref{eq:closedloop} satisfies ${\bx(t) \in \C}$ for all time.

Let safety be described by a continuously differentiable function ${h : \X \to \R}$, along with the safe set $\C$ defined as:
\begin{equation} \label{eq:safeset}
    \C = \{\bx \in \X: h(\bx) \geq 0 \}.
    % \\
    % \partial \C & = \{\bx \in \X: h(\bx) = 0 \}, \\
    % {\rm Int} \C & = \{\bx \in \X: h(\bx) > 0 \}.
\end{equation}
That is, the state $\bx$ is safe if ${h(\bx) \geq 0}$.
If $h$ has the properties of a CBF, safe controllers can be synthesized for~\eqref{eq:system}.

\begin{definition}[\cite{AmesXuGriTab2017}]
Function $h$ is a {\em control barrier function} for~\eqref{eq:system} on $\C$ if there exists\footnote{Function ${\alpha : (-b,a) \to \R}$, ${a,b>0}$ is of extended class-$\K$ (${\alpha \in \Ke}$) if it is continuous, strictly increasing, and ${\alpha(0)=0}$.}
% Function ${\alpha : \R \to \R}$ is of extended class-$\Kinf$ (${\alpha \in \Keinf}$) if ${\alpha \in \Ke}$ and ${\lim_{r \to \pm \infty} \alpha(r) = \pm \infty}$.}
${\alpha \in \Ke}$ such that for all ${\bx \in \C}$:
\begin{equation} \label{eq:CBF_condition}
    \sup_{\bu \in \U} \dot{h}(\bx,\bu) > - \alpha \big( h(\bx) \big),
\end{equation}
where ${\dot{h}(\bx,\bu) = \derp{h}{\bx}(\bx) \cdot \big( \bff(\bx) + \bg(\bx) \bu \big)}$.
% where $\dot{h}$ denotes the derivative of $h$ along~\eqref{eq:system}:
% \begin{equation} \label{eq:hdot}
%     \dot{h}(\bx,\bu) =
%     \derp{h}{\bx}(\bx) \cdot \big( \bff(\bx) + \bg(\bx) \bu \big).
% \end{equation}
\end{definition}

\begin{theorem}[\cite{AmesXuGriTab2017}] \label{theo:CBF}
\textit{
If $h$ is a CBF for~\eqref{eq:system} on $\C$, then any locally Lipschitz controller ${\bk : \X \to \U}$ that satisfies: 
\begin{equation} \label{eq:safety_condition}
    \dot{h} \big( \bx, \bk(\bx) \big) \geq - \alpha \big( h(\bx) \big)
\end{equation}
for all ${\bx \in \C}$ renders~\eqref{eq:closedloop} safe w.r.t.~$\C$.
% Furthermore, if~\eqref{eq:safety_condition} holds for all ${\bx \in \X}$, then $\C$ is asymptotically stable.
}
\end{theorem}

Safe controllers can be synthesized by enforcing~\eqref{eq:safety_condition} during control design.
For example,~\eqref{eq:safety_condition} can be used as constraint in optimization to minimally modify an existing desired controller ${\bk_{\rm d} : \X \to \U}$ and generate a safe controller:
\begin{equation} \label{eq:QP}
\begin{aligned}
    \bk(\bx) = \underset{\bu \in \U}{\operatorname{argmin}} & \quad \| \bu - \bk_{\rm d}(\bx) \|^2 \\[-3pt]
    \text{s.t.} & \quad \dot{h}(\bx,\bu) \geq - \alpha \big( h(\bx) \big).
\end{aligned}
\end{equation}
This optimization problem can be solved in closed form~\cite{cohen2023smooth}.
% which can be written equivalently in explicit form~\cite{cohen2023smooth}:
% \begin{equation} \label{eq:safetyfilter}
%     \bk(\bx) = \bk_{\rm d}(\bx) + \Lambda \big( a(\bx), \| \bb(\bx) \| \big) \bb(\bx)^\top,
% \end{equation}
% with:
% \begin{align}
%     & a(\bx) = \dot{h} \big( \bx, \bk_{\rm d}(\bx) \big) + \alpha \big( h(\bx) \big), \quad
%     \bb(\bx) = \derp{h}{\bx}(\bx) \cdot \bg(\bx), \nonumber \\
%     & \Lambda(a,b) =
%     \begin{cases}
%         0 & {\rm if}\ b = 0, \\
%         \frac{1}{b} \max \Big\{ 0, - \frac{a}{b} \Big\} & {\rm if}\ b \neq 0.
%     \end{cases}
% \end{align}

\begin{remark} \label{rem:timedependency}
If safety constraints depend on time, the CBF ${h : \X \times \R \to \R}$ and controller ${\bk : \X \times \R \to \U}$ also need to be time dependent, whereas~\eqref{eq:safety_condition} needs to be modified to:
\begin{equation} \label{eq:safety_condition_time_varying}
    \dot{h} \big( \bx, t, \bk(\bx,t) \big) \geq - \alpha \big( h(\bx,t) \big),
\end{equation}
where ${\dot{h}(\bx,t,\bu) = \frac{\partial h}{\partial t}(\bx,t) + \derp{h}{\bx}(\bx,t) \cdot \big( \bff(\bx) + \bg(\bx) \bu \big)}$.
% where the derivative of $h$ changes from~\eqref{eq:hdot} to:
% \begin{equation} \label{eq:hdot_time_varying}
%     \dot{h}(\bx,t,\bu) =
%     \frac{\partial h}{\partial t}(\bx,t) + \derp{h}{\bx}(\bx,t) \cdot \big( \bff(\bx) + \bg(\bx) \bu \big).
% \end{equation}
\end{remark}

This paper addresses control design for polytope safe sets.
\begin{definition}
Consider the set $\C$ in~\eqref{eq:safeset}.
$\C$ is called a {\em half space} if $h$ is affine in $\bx$.
$\C$ is called a {\em convex polytope} if it is the intersection of finitely many half spaces.
$\C$ is called a {\em polytope} if it is the union of finitely many convex polytopes. 
\end{definition}

This necessitates set compositions (unions and intersections).
Consider $N$ sets denoted by $\C_{i}$ with functions $h_{i}$ and index ${i \in I = \{1, 2, \ldots, N\}}$.
The union and intersection are given by the $\max$ and $\min$ functions, respectively,~\cite{glotfelter2017nonsmooth, molnar2023composing}:
% on the individual barriers $h_{i}$:
\begin{align}
    {\textstyle \bigcup_{i \in I}} \C_{i} = \Big\{ \bx \in \X: \max_{i \in I} h_{i}(\bx) \geq 0 \Big\},
    \label{eq:maxcbf} \\
    {\textstyle \bigcap_{i \in I}} \C_{i} = \Big\{ \bx \in \X: \min_{i \in I} h_{i}(\bx) \geq 0 \Big\}.
    \label{eq:mincbf}
\end{align}
% These compositions will be relevant below, where we obtain convex polytopes as the intersections of half spaces and general polytopes as the unions of convex polytopes.

Via smooth approximations of the $\max$ and $\min$ functions, a continuously differentiable CBF candidate can be obtained such that
${h(\bx) \approx \max_{i \in I} h_{i}(\bx)}$ and
${h(\bx) \approx \min_{i \in I} h_{i}(\bx)}$.
For example, the log-sum-exp formulas from~\cite{molnar2023composing, lindemann2019stl}:
\begin{equation} \label{eq:logsumexp}
    \max_{i \in I} a_{i} \approx \frac{1}{\kappa} \ln \!\bigg(\! \sum_{i \in I} {\rm e}^{\kappa a_{i}} \!\bigg), \quad
    \min_{i \in I} b_{i} \approx - \frac{1}{\kappa} \ln \!\bigg(\! \sum_{i \in I} {\rm e}^{-\kappa b_{i}} \!\bigg),
\end{equation}
offer one possible option for smoothing, for any ${a_{i}, b_{i} \in \R}$, ${i \in I}$.
These formulas over-approximate the $\max$ and under-approximate the $\min$ function, and ${\kappa>0}$ is a smoothing parameter that determines the approximation error so that ${\kappa \to \infty}$ recovers the $\max$ and $\min$.

Finally, in case of bounded polytopes we rely on vertex representations, and we will use the notion of a convex hull that is the smallest convex polytope containing $\C$.
\begin{definition} \label{def:hull}
Let $\C$ be a bounded polytope with $N_{\rm v}$ vertices denoted by $\bx_{k} \in \C$, ${k \in K = \{1, 2, \ldots, N_{\rm v}}\}$.
The {\em convex hull} of the vertices is defined by:
\begin{equation}
    \H = \bigg\{ \bx \in \X : \bx = \sum_{k \in K} \lambda_{k} \bx_{k} \bigg\},
\end{equation}
where ${\lambda_{k} \in [0,1]}$ for all ${k \in K}$ and ${\sum_{k \in K} \lambda_{k} = 1}$.
\end{definition}
% Note that if $\C$ is a convex polytope, the convex hull is $\C$ itself.
% If $\C$ is a polytope but not a convex one, then the convex hull is the smallest convex polytope that contains $\C$.

%%%%%%%%%%%%%%%%%%%%%%%%%%%%%%%%%%%%%%%%%%%%%%%%%%%%%%%%%%%%%%%%%%%%%%%%%%%%%%%%
\section{NAVIGATION IN POLYTOPE ENVIRONMENTS}

In this paper, we use CBF theory for safety-critical navigation and achieve collision-free motions for a controlled agent in an environment.
The geometries of both the agent and its environment are described as polytopes.
We propose a method to construct CBF candidates for polytope geometries.

We formulate CBF candidates in terms of the agent's position ${\bp \in \R^p}$, with ${p=2}$ for 2D and ${p=3}$ for 3D navigation problems.
For example, for single integrator dynamics:
\begin{equation} \label{eq:single_integrator}
    \dot{\bx} = \bu,
\end{equation}
the state $\bx$ matches the position $\bp$ (i.e., ${\bx = \bp}$), and the proposed CBF candidate can be directly used to synthesize a safe velocity input ${\bu = \bk(\bx)}$.
In more complex models, the position $\bp$ is typically a member of the state $\bx$.
% which may contain other variables as well (such as orientation, velocity, angular velocity).

\begin{remark} \label{rem:robotic}
To focus on the geometry of the problem, we use the single integrator in our examples, while our framework also applies to robotic systems.
Tracking the safe velocity of the single integrator by low-level controllers can guarantee safe obstacle avoidance on various robotic systems, including manipulators, legged, wheeled, and flying robots~\cite{singletary2022manipulation, molnar2022modelfree, cohen2024reduced}.
Furthermore, in second-order robotic systems:
\begin{equation}
\begin{aligned}
\dot{\bq} & = \bv, \\
\dot{\bv} & = \bD(\bq)^{-1} \big( -\bC(\bq,\bv) \bv - \bG(\bq) + \bB \bu \big),
\end{aligned}
\end{equation}
a CBF $h$ for the single integrator can be extended to a valid CBF $H$ for the second-order system using, e.g., backstepping
as ${H(\bx) = h(\bq) - \|\bv - \bk(\bq)\|^2/(2\mu)}$ with the safe velocity ${\bk(\bq)}$ of the single integrator and $\mu>0$~\cite{cohen2024reduced, taylor2022safebackstepping, cohen2024constructive}.
\end{remark}

\subsection{Navigation of Point Agent}

We begin with the safety-critical navigation of an agent modeled as a point.
We describe the agent's safety based on its position $\bp$ and the geometry of the environment. First, we discuss the simplest environment: a half space bounded by a single wall.
Then, we generalize to convex polytope environments as intersections of half spaces.
Finally, we address general polytope environments as unions of convex polytopes.
Polytopes in the position space are formulated as:
\begin{equation} \label{eq:polytope}
    \P = \{ \bp \in \R^p: \psi(\bp) \geq 0\},
\end{equation}
where use the notation ${\psi : \R^p \to \R}$ for constraint functions that are not necessarily CBFs.
The underlying CBF candidates defined over the state space will be denoted as $h$.

\subsubsection{Half Space Environment}

A safe half space bounded by a single barrier (i.e., a wall) can be described by:
\begin{equation} \label{eq:barrier_point_wall}
    \psi(\bp) = \bn \cdot (\bp - \bw),
\end{equation}
where ${\bn \in \R^p}$, ${\bn \neq \bzero}$, is the normal vector and ${\bw \in \R^p}$ is the location of the barrier (wall), whereas ${\psi(\bp) \geq 0}$ indicates that the agent is on the safe side of the barrier.
For the single integrator in~\eqref{eq:single_integrator}, this directly yields the CBF ${h(\bx)=\psi(\bp)}$.
% Note that the environment can be in 2D space (${\bp,\bn,\bw \in \R^2}$) or in 3D space (${\bp,\bn,\bw \in \R^3}$).

% \begin{equation}
%     \frac{\partial \psi}{\partial \bp}(\bp) = \bn,
% \end{equation}

% \begin{remark} Moving (i.e., translating and rotating) walls can also be captured by the time-dependent barrier:
% \begin{equation}
%     \psi(\bp,t) = \bn(t) \cdot \big( \bp - \bw(t) \big).
% \end{equation}
% % e.g. rotating wall:
% % \begin{equation}
% %     \bn(t) = \bR(\varphi(t)) \bn_{\rm w}, \quad
% %     \bR(\varphi(t)) =
% %     \begin{bmatrix}
% %     \cos (\varphi(t)) \!&\! -\sin (\varphi(t)) \\
% %     \sin (\varphi(t)) \!&\! \cos (\varphi(t))
% %     \end{bmatrix},
% % \end{equation}
% % which has the derivatives:
% % \begin{equation}
% % \begin{aligned}
% %     \frac{\partial \psi}{\partial \bp}(\bp,t) & = \bn(t), \\
% %     \frac{\partial \psi}{\partial t}(\bp,t) & = \frac{\partial \bn}{\partial t}(t) \cdot \big( \bp - \bw(t) \big) - \bn(t) \cdot \frac{\partial \bw}{\partial t}(t).
% % \end{aligned}
% % \end{equation}
% For simplicity, we will discuss the time-independent case and drop $t$.
% Time dependency, however, can be added throughout the paper by taking into account ${\partial \psi / \partial t}$; see Remark~\ref{rem:timedependency}.
% \end{remark}

\subsubsection{Convex Polytope Environment}

Let the environment be a convex polytope that is the intersection of $N_{\rm w}$ half spaces, with index ${i \in I = \{1, 2, \ldots, N_{\rm w}\}}$,
% Each half space has a
normal vector ${\bn_{i} \in \R^p}$, ${\bn_{i} \neq \bzero}$, location ${\bw_{i} \in \R^p}$, and barrier:
\begin{equation} \label{eq:barrier_point_wall_multiple}
    \psi_{i}(\bp) = \bn_{i} \cdot (\bp - \bw_{i}).
\end{equation}
To describe the intersection, the individual barriers can be combined into a CBF candidate according to~\eqref{eq:mincbf}:
\begin{equation} \label{eq:barrier_point_convex_polytope}
    \psi(\bp) = \min_{i \in I} \psi_{i}(\bp).
\end{equation}
The agent does not collide with the environment if ${\psi(\bp) \geq 0}$.

Due to the $\min$ function, $\psi$ may not be differentiable, and hence ${h(\bx)=\psi(\bp)}$ is not a CBF.
To obtain a continuously differentiable CBF candidate, a smooth approximation of the $\min$ function can be used, for example, based on~\eqref{eq:logsumexp}:
\begin{equation} \label{eq:smoothing_point_convex_polytope}
    h(\bx) = -\frac{1}{\kappa} \ln \bigg( \sum_{i \in I} {\rm e}^{-\kappa \psi_{i}(\bp)} \bigg).
\end{equation}
\begin{remark} \label{rem:smoothing}
While nonsmooth barrier function theory~\cite{glotfelter2017nonsmooth} could enable the direct use of~\eqref{eq:barrier_point_convex_polytope}, a smooth $h$ in~\eqref{eq:smoothing_point_convex_polytope} facilitates the tracking of the safe velocity and the CBF extension for robotic systems~\cite{cohen2024reduced} mentioned in Remark~\ref{rem:robotic}.
The smooth approximation makes the safe region slightly smaller by smoothing the corners of the environment, with a trade-off between smoothness and conservativeness.
Increasing $\kappa$ can reduce the approximation error and conservativeness without affecting computation times, while it leads to a more rapidly changing gradient for $h$~\cite{molnar2023composing}.
The nonsmooth case with sharp corners and discontinuous gradient is recovered for ${\kappa \to \infty}$.
\end{remark}

\subsubsection{General Polytope Environment}
Let the environment be a general polytope that is the union of $N_{\rm p}$ convex polytopes with index ${j \in J = \{1, 2, \ldots, N_{\rm p}\}}$, constructed from a total of $N_{\rm w}$ half spaces with index ${i \in I}$.
Let ${I_{j} \subseteq I}$ indicate those half spaces that form convex polytope $j$ (i.e., half space $i$ belongs to convex polytope $j$ if ${i \in I_{j}}$).
The union of the convex polytopes can be described by the following composition of the individual barriers based on~\eqref{eq:maxcbf} and~\eqref{eq:barrier_point_convex_polytope}:
\begin{equation} \label{eq:barrier_point_general_polytope}
    \psi(\bp) = \max_{j \in J} \min_{i \in I_{j}} \psi_{i}(\bp).
\end{equation}

This can be smoothly approximated, for example, via~\eqref{eq:logsumexp}:
\begin{equation} \label{eq:smoothing_point_general_polytope}
    h(\bx) = \frac{1}{\kappa} \ln \bigg( \sum_{j \in J} \Big( \sum_{i \in I_{j}} {\rm e}^{-\kappa \psi_{i}(\bp)} \Big)^{-1} \bigg) - \frac{b}{\kappa}.
\end{equation}
Here ${b>0}$ is a buffer parameter which can guarantee safety even with approximation errors, i.e., ensure that $h(\bx)$ is an under-approximation (${\psi(\bp) \geq h(\bx)}$) so that ${\bx \in \C}$ implies ${\bp \in P}$.
Note that $b$ is needed because the $\max$ function is over-approximated by~\eqref{eq:logsumexp} (while it was not needed in~\eqref{eq:smoothing_point_convex_polytope} as the $\min$ is under-approximated).
One may choose $b$ and $\kappa$, e.g., based on the approximation error bound in~\cite[Thm.~5]{molnar2023composing}.

\begin{remark} Our framework can also describe dynamically changing polytope environments via time-dependent barriers:
\begin{equation}
    \psi_{i}(\bp,t) = \bn_{i}(t) \cdot \big( \bp - \bw_{i}(t) \big),
\end{equation}
% e.g. rotating wall:
% \begin{equation}
%     \bn(t) = \bR(\varphi(t)) \bn_{\rm w}, \quad
%     \bR(\varphi(t)) =
%     \begin{bmatrix}
%     \cos (\varphi(t)) \!&\! -\sin (\varphi(t)) \\
%     \sin (\varphi(t)) \!&\! \cos (\varphi(t))
%     \end{bmatrix},
% \end{equation}
% which has the derivatives:
% \begin{equation}
% \begin{aligned}
%     \frac{\partial \psi}{\partial \bp}(\bp,t) & = \bn(t), \\
%     \frac{\partial \psi}{\partial t}(\bp,t) & = \frac{\partial \bn}{\partial t}(t) \cdot \big( \bp - \bw(t) \big) - \bn(t) \cdot \frac{\partial \bw}{\partial t}(t).
% \end{aligned}
% \end{equation}
with translation and rotation for each boundary, 
as long as their composition~\eqref{eq:barrier_point_general_polytope} (i.e., the topology given by $I_{j}$, $J$) is time-independent.
Time-dependency can be added throughout the paper by accounting for ${\partial \psi / \partial t}$ as in Remark~\ref{rem:timedependency}.
\end{remark}

\begin{figure}
\centering
\includegraphics[scale=1]{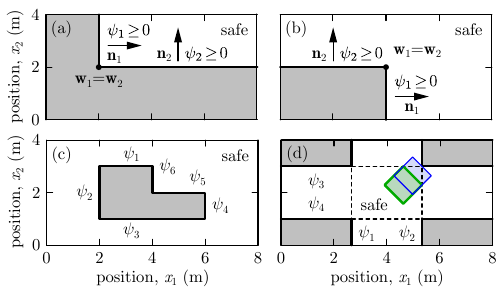}
\vspace{-3mm}
\caption{
Examples of polytope environments and the corresponding barriers: (a) convex corner, (b) concave corner, (c) L-shaped obstacle, (d) crossroad.
}
\vspace{-5mm}
\label{fig:sets}
\end{figure}

\begin{example} \label{ex:sets}
We illustrate the barriers of polytope environments via toy examples.
We drop the arguments of $\psi$ and $\psi_{i}$ for brevity.
The convex corner in Fig.~\ref{fig:sets}(a) is captured by:
\begin{equation}
    \psi = \min \{ \psi_{1}, \psi_{2} \},
\end{equation}
% cf.~\eqref{eq:barrier_point_general_polytope}
with
% ${J = \{1\}}$,
% ${I_{1} = I = \{1,2\}}$,
${\bn_{1} = \begin{bmatrix} 1 \!&\! 0 \end{bmatrix}^\top}$,
${\bn_{2} = \begin{bmatrix} 0 \!&\! 1 \end{bmatrix}^\top}$, and
${\bw_{1} = \bw_{2} = \begin{bmatrix} 2 \!&\! 2 \end{bmatrix}^\top}$.
Similarly, the concave corner in Fig.~\ref{fig:sets}(b) is given by:
\begin{equation}
    \psi = \max \{ \psi_{1}, \psi_{2} \}.
\end{equation}
% with
% ${J = \{1,2\}}$,
% ${I_{1} = \{1\}}$,
% ${I_{2} = \{2\}}$, and
% ${I = \{1,2\}}$.
% \begin{equation}
%     h = \max \{ \psi_{1}, \psi_{2}, \psi_{3}, \psi_{4}\},
% \end{equation}
% \begin{equation}
%     \!\!\!\!\!\!
%     h = \min \!\big\{\!
%     \max \{ \psi_{1}, \psi_{2}, \psi_{3}, \psi_{4} \},
%     \max \{ \psi_{5}, \psi_{6}, \psi_{7}, \psi_{8} \}
%     \big\},
% \end{equation}

The L-shaped object in Fig.~\ref{fig:sets}(c) consists of convex corners (related to $\min$) and concave corners (associated with $\max$):
\begin{equation} \label{eq:barrier_L_shape}
    \psi = \max \big\{ \psi_{1}, \psi_{2}, \psi_{3}, \psi_{4}, \min\{ \psi_{5}, \psi_{6} \} \big\},
\end{equation}
which is constructed according to~\eqref{eq:barrier_point_general_polytope} from ${N_{\rm w}=6}$ barriers and ${N_{\rm p}=5}$ convex polytopes with boundaries
% ${J = \{1,2,3,4,5\}}$,
${I_{1} = \{1\}}$,
${I_{2} = \{2\}}$,
${I_{3} = \{3\}}$,
${I_{4} = \{4\}}$, and
${I_{5} = \{5,6\}}$.
% ${I = \{1,2,3,4,5,6\}}$.

The crossroad in Fig.~\ref{fig:sets}(d) is the union of two roads, each of which is a convex polytope with two boundaries:
\begin{equation}
    \psi = \max \!\big\{\!
    \min \{ \psi_{1}, \psi_{2} \},
    \min \{ \psi_{3}, \psi_{4} \}
    \big\},
\end{equation}
cf.~\eqref{eq:barrier_point_general_polytope} with
${N_{\rm w}=4}$,
${N_{\rm p}=2}$,
${I_{1} = \{1,2\}}$, and
${I_{2} = \{3,4\}}$.
\end{example}

\begin{example} \label{ex:L_shape}
We simulate an obstacle avoidance task where a point agent, governed by~\eqref{eq:single_integrator}, aims to reach a goal point while avoiding the L-shaped obstacle from Fig.~\ref{fig:sets}(c).
The agent uses the controller~\eqref{eq:QP} that modifies a desired velocity $\bk_{\rm d}(\bx)$ to a safe velocity ${\bu = \bk(\bx)}$ using the CBF~\eqref{eq:smoothing_point_general_polytope} that approximates~\eqref{eq:barrier_L_shape}, with
${\kappa = 5}$,
${b = 0.7}$, and
${\alpha(h) = 2 h}$.
To reach the goal at point $\bp_{\rm g}$, the desired controller is:
\begin{equation}
\bk_{\rm d}(\bx)={\rm sat} \big( K_{\rm p} (\bp_{\rm g} - \bp) \big),
\end{equation}
with gain ${K_{\rm p} = 1}$ and saturation at ${u_{\max} = 1}$ (${{\rm sat}(\bu)=\bu}$ if ${\|\bu\| \leq u_{\max}}$ and ${{\rm sat}(\bu)=\frac{\bu}{\|\bu\|} u_{\rm max}}$ if ${\|\bu\| > u_{\max}}$).

The simulation results\footnote{Matlab codes for each simulation example are available at: https://github.com/molnartamasg/CBFs-for-polytope-navigation.} are shown in Fig.~\ref{fig:L_shape}.
The thick trajectory in panel (a) highlights that the agent successfully reaches the goal without collision.
Safety is maintained throughout the motion as indicated by the nonnegative value of the CBF in panel (b).
This is achieved by modifying the desired inputs (thin) to safe inputs (thick) in panel (c).
Moreover, the thin trajectories in panel (a) show that collision free-motions are also achieved for other starting positions.
\end{example}

\begin{figure}
\centering
\includegraphics[scale=1]{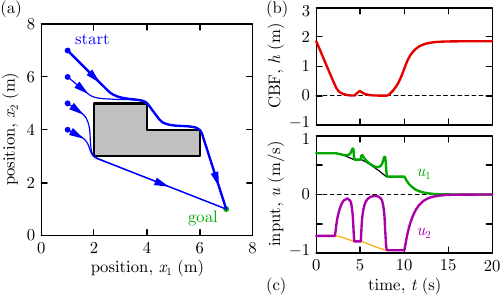}
\vspace{-3mm}
\caption{
Safety-critical navigation of a point agent around a polytope obstacle.
Safety is maintained using the proposed CBF candidate~\eqref{eq:smoothing_point_general_polytope}.
}
\vspace{-5mm}
\label{fig:L_shape}
\end{figure}

\subsection{Navigation of Polytope Agent}

Next we describe the safety-critical navigation of an agent modeled as a rigid body.
The agent's geometry is assumed to be a bounded polytope with $N_{\rm v}$ vertices indexed by ${k \in K = \{1, 2, \ldots, N_{\rm v}\}}$.
The positions of the vertices are:
\begin{equation}
    \bp_{k} = \bp + \Delta \bp_{k},
\end{equation}
with the relative position $\Delta \bp_{k}$ from the agent's center at $\bp$.

\begin{remark} \label{rem:rotations}
If the agent moves without rotation, which is the case for the single integrator~\eqref{eq:single_integrator}, then $\Delta \bp_{k}$ is constant parameter.
If the agent rotates, then $\Delta \bp_{k}$ becomes state-dependent as it depends on the agent's orientation which needs to be represented in the state $\bx$ (e.g. via Euler angles).
% With smoothing, like in~\eqref{eq:smoothing_point_convex_polytope} and~\eqref{eq:barrier_point_general_polytope}, the differentiability of the resulting CBF candidate w.r.t.~the agent's orientation can be ensured.
\end{remark}

% \begin{remark}
% If the agent is rotating in 2D plane, $\bp_{k}$ also depends on the orientation angle $\chi$:
% \begin{equation}
%     \bp_{k} = \bp + \bR(\chi) \Delta \bp_{k}, \quad
%     \bR(\chi) =
%     \begin{bmatrix}
%     \cos \chi & -\sin \chi \\
%     \sin \chi & \cos \chi
%     \end{bmatrix}.
% \end{equation}
% If the agent rotates in 3D space, $\bp_{k}$ depends on the Euler angles (roll, pitch, yaw) $\phi$, $\theta$, $\chi$:
% \begin{align}
%     & \bp_{k} = \bp + \bR(\phi,\theta,\chi) \Delta \bp_{k}, \\
%     & \bR(\phi,\theta,\chi) =
%     % \begin{bmatrix}
%     %     \cos \chi \!&\! - \sin \chi \!&\! 0 \\
%     %     \sin \chi \!&\! \cos \chi \!&\! 0 \\
%     %     0 \!&\! 0 \!&\! 1
%     % \end{bmatrix}
%     % \begin{bmatrix}
%     %     \cos \theta \!&\! 0 \!&\! \sin \theta \\
%     %     0 \!&\! 1 \!&\! 0 \\
%     %     - \sin \theta \!&\! 0 \!&\! \cos \theta
%     % \end{bmatrix},
%     \begin{bmatrix}
%         {\rm c}_{\chi} \!&\! - {\rm s}_{\chi} \!&\! 0 \\
%         {\rm s}_{\chi} \!&\! {\rm c}_{\chi} \!&\! 0 \\
%         0 \!&\! 0 \!&\! 1
%     \end{bmatrix} \!\!
%     \begin{bmatrix}
%         {\rm c}_{\theta} \!&\! 0 \!&\! {\rm s}_{\theta} \\
%         0 \!&\! 1 \!&\! 0 \\
%         - {\rm s}_{\theta} \!&\! 0 \!&\! {\rm c}_{\theta}
%     \end{bmatrix} \!\!
%     \begin{bmatrix}
%         1 \!&\! 0 \!&\! 0 \\
%         0 \!&\! {\rm c}_{\phi} \!&\! - {\rm s}_{\phi} \\
%         0 \!&\! {\rm s}_{\phi} \!&\! {\rm c}_{\phi}
%     \end{bmatrix}\!, \nonumber
% \end{align}
% where ${\rm c}_{(.)}$ abbreviates $\cos(.)$ and ${\rm s}_{(.)}$ abbreviates $\sin(.)$.
% \end{remark}

We describe the safety of such polytope agents in half space, convex polytope, and general polytope environments.

\subsubsection{Half Space Environment}
Consider the single half space (wall) from~\eqref{eq:barrier_point_wall}.
The polytope agent is safe if all its vertices are in this half space, i.e.,  ${\bp_{k} \in \P}$ for all ${k \in K}$ with $\P$ given by~\eqref{eq:polytope} and~\eqref{eq:barrier_point_wall}.
Using the $\min$ function, we can construct a function $\phi$ so that ${\phi(\bx) \geq 0}$ indicates safety:
\begin{equation} \label{eq:barrier_body_wall}
    \phi(\bx) = \min_{k \in K} \psi(\bp_{k}).
\end{equation}
% We use the notation $\phi$ rather than $h$ to emphasize that $\phi$ may not be a CBF due to non-differentiability.
Note that $\phi$ may depend on the state $\bx$ rather than just the position $\bp$ because $\bp_{k}$ may be orientation-dependent if the agent rotates; cf.~Remark~\ref{rem:rotations}.
Equation~\eqref{eq:barrier_body_wall} can be interpreted as the agent's center must be located in a convex polytope whose half spaces are obtained by shifting the wall by $\Delta \bp_{k}$.
A smooth counterpart of~\eqref{eq:barrier_body_wall} can be obtained similar to~\eqref{eq:smoothing_point_convex_polytope}.

\subsubsection{Convex Polytope Environment}
If the environment is given by a convex polytope, as it was described in~\eqref{eq:barrier_point_convex_polytope}, the agent is safe if all its vertices are in this polytope.
Similar to~\eqref{eq:barrier_body_wall}, this can be expressed using the $\min$ function:
\begin{equation} \label{eq:barrier_body_convex_polytope}
    \phi(\bx) = \min_{i \in I} \min_{k \in K} \psi_{i}(\bp_{k}),
\end{equation}
where we note that the order of the two $\min$ operators is interchangeable.
A smooth counterpart is:
\begin{equation} \label{eq:smoothing_body_convex_polytope}
    h(\bx) = -\frac{1}{\kappa} \ln \bigg( \sum_{i \in I} \sum_{k \in K} {\rm e}^{-\kappa \psi_{i}(\bp_{k})} \bigg).
    % - \frac{b}{\kappa},
\end{equation}
% where even ${b=0}$ may guarantee safety since this is an under-approximation of~\eqref{eq:barrier_body_convex_polytope}.

\subsubsection{General Polytope Environment}
If the environment is a general (not necessarily convex) polytope, safety could be captured by a complicated signed distance expression.
This expression can be simplified significantly by the following under-approximation.
We use a sufficient condition for safety: the agent is located in the polytope environment if all its vertices are in one of the convex polytopes making up the environment.
Mathematically, this is expressed by:
\begin{equation} \label{eq:barrier_body_general_polytope}
    \phi(\bx) = \max_{j \in J} \min_{i \in I_{j}} \min_{k \in K} \psi_{i}(\bp_{k}),
\end{equation}
% cf.~\eqref{eq:barrier_point_general_polytope},
where the order of the $\max$ and $\min$ operators is important and not interchangeable.
This formula simplifies to~\eqref{eq:barrier_point_general_polytope} if the agent is a point (${K=\{1\}}$) and it simplifies to~\eqref{eq:barrier_body_convex_polytope} if the environment is a convex polytope (${J=\{1\}}$).

This under-approximation of the signed distance is illustrated in Fig.~\ref{fig:sets}(d) for a diamond-shaped agent.
The thin blue lines show a position with zero signed distance.
The agent touches the corner here.
The thick green lines highlight a position with ${\phi(\bx)=0}$.
The agent touches the boundaries of convex polytopes (roads) but not the corner.
The buffer around the corner illustrates how conservative~\eqref{eq:barrier_body_general_polytope} is.

The fact that ${\phi(\bx) \geq 0}$ provides collision-free motion is stated formally by the following theorem.
\begin{theorem}
\textit{
Consider a bounded polytope agent with vertices ${\bp_{k} \in \R^p}$ (${k \in K}$) and convex hull $\H$, a polytope environment $\P$ given by~\eqref{eq:polytope} and~\eqref{eq:barrier_point_general_polytope}, and function $\phi$ defined in~\eqref{eq:barrier_body_general_polytope}.
If ${\phi(\bx) \geq 0}$, then ${\H \subseteq \P}$ holds, i.e., the convex hull of the agent is safely contained in the polytope environment.
}
\end{theorem}

\begin{proof}
Based on Definition~\ref{def:hull} and~\eqref{eq:polytope}, ${\H \subseteq \P}$ holds if:
\begin{equation} \label{eq:hull_in_polytope}
    \psi \Big( \sum_{k \in K} \lambda_{k} \bp_{k} \Big)
    \geq 0,
\end{equation}
for all ${\lambda_{k} \in [0,1]}$, ${k \in K}$, such that ${\sum_{k \in K} \lambda_{k} = 1}$.
We prove that~\eqref{eq:hull_in_polytope} holds for ${\phi(\bx) \geq 0}$ by showing that:
\begin{equation} \label{eq:convex_hull_approx}
    \psi \Big( \sum_{k \in K} \lambda_{k} \bp_{k} \Big) \geq \phi(\bx).
\end{equation}
Based on~\eqref{eq:barrier_point_wall_multiple}, ${\sum_{k \in K} \lambda_{k} = 1}$, and ${\lambda_{k} \in [0,1]}$, we have:
\begin{equation}
    \psi_{i} \Big( \sum_{k \in K} \lambda_{k} \bp_{k} \Big)
    = \sum_{k \in K} \lambda_{k} \psi_{i}(\bp_{k})
    \geq \min_{k \in K} \psi_{i}(\bp_{k}),
\end{equation}
for all ${i \in I}$.
% Because ${\sum_{k \in K} \lambda_{k} = 1}$, the following identity holds:
% \begin{equation}
%     \bn_{i} \cdot \Big( \sum_{k \in K} \lambda_{k} \bp_{k} - \bw_{i} \Big)
%     = \sum_{k \in K} \lambda_{k} \bn_{i} \cdot (\bp_{k} - \bw_{i}),
% \end{equation}
% for all ${i \in I}$.
% Based on~\eqref{eq:barrier_point_wall_multiple}, this is equivalent to:
% \begin{equation}
%     \psi_{i} \Big( \sum_{k \in K} \lambda_{k} \bp_{k} \Big)
%     = \sum_{k \in K} \lambda_{k} \psi_{i}(\bp_{k}).
% \end{equation}
% Because ${\sum_{k \in K} \lambda_{k} = 1}$ and ${\lambda_{k} \in [0,1]}$, this leads to:
% \begin{equation}
%     \psi_{i} \Big( \sum_{k \in K} \lambda_{k} \bp_{k} \Big)
%     \geq \min_{k \in K} \psi_{i}(\bp_{k}),
% \end{equation}
% for all ${i \in I}$.
It can be shown that this ultimately yields:
\begin{equation}
    \max_{j \in J} \min_{i \in I_{j}} \psi_{i} \Big( \sum_{k \in K} \lambda_{k} \bp_{k} \Big)
    \geq \max_{j \in J} \min_{i \in I_{j}} \min_{k \in K} \psi_{i}(\bp_{k}),
\end{equation}
which is equivalent to~\eqref{eq:convex_hull_approx} according to~\eqref{eq:barrier_point_general_polytope} and~\eqref{eq:barrier_body_general_polytope}.
\end{proof}

\begin{figure}
\centering
\includegraphics[scale=1]{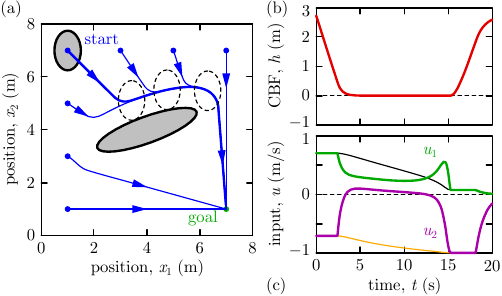}
\vspace{-3mm}
\caption{
Safety-critical navigation of an ellipse agent around an ellipse obstacle.
Safety is maintained by approximating both ellipses as polytopes and using the proposed CBF candidate~\eqref{eq:smoothing_body_general_polytope}.
}
\vspace{-5mm}
\label{fig:ellipse}
\end{figure}

% The function $\phi$ in~\eqref{eq:barrier_body_general_polytope} indicates the agent's safety if ${\phi(\bx) \geq 0}$.
A smooth approximation of the $\max$ and $\min$ functions in~\eqref{eq:barrier_body_general_polytope} yields a continuously differentiable CBF candidate.
For example, the log-sum-exp formula~\eqref{eq:logsumexp} gives:
\begin{equation} \label{eq:smoothing_body_general_polytope}
    h(\bx) = \frac{1}{\kappa} \ln \bigg( \sum_{j \in J} \Big( \sum_{i \in I_{j}} \sum_{k \in K} {\rm e}^{-\kappa \psi_{i}(\bp_{k})} \Big)^{-1} \bigg) - \frac{b}{\kappa}.
\end{equation}
% where $\kappa$ is a smoothing parameter and $b$ is a buffer.
We remark that using ${b=0}$ may be a sufficient buffer for collision-free motion because~\eqref{eq:barrier_body_general_polytope} is an under-approximation of the signed distance which already works as a buffer; cf.~Fig.~\ref{fig:sets}(d).
% around the corners of the environment.
Importantly,~\eqref{eq:smoothing_body_general_polytope} is a simple closed form expression.
Thus, $h$ and its gradient can be calculated in real time via elementary operations, including the evaluation of ${N_{\rm w} \times N_{\rm v}}$ linear functions, their exponentials, sums, and a logarithm.
Such simplicity is especially relevant when extending CBFs to more complex systems; cf.~Remark~\ref{rem:robotic}.

% The corresponding gradient of $h$ is:
% \begin{equation} \label{eq:grad_body_general_polytope}
%     \derp{h}{\bx}(\bx) = \frac{1}{{\rm e}^{\kappa h(\bx) + b}} \sum_{j \in J} \frac{\sum_{i \in I_{j}} \sum_{k \in K} {\rm e}^{-\kappa \psi_{i}(\bp_{k})} \bn_{i} \cdot \derp{\bp_{k}}{\bx}}{\Big( \sum_{i \in I_{j}} \sum_{k \in K} {\rm e}^{-\kappa \psi_{i}(\bp_{k})} \Big)^{2}} .
% \end{equation}
% Equations~\eqref{eq:smoothing_body_general_polytope}-\eqref{eq:grad_body_general_polytope} can be computed efficiently using matrix products with ${\bPsi \in \R^{N_{\rm w} \times N_{\rm v}}}$ and ${\bA \in \R^{N_{\rm w} \times N_{\rm p}}}$ defined by:
% \begin{equation}
%     \Psi_{ik} = {\rm e}^{-\kappa \psi_{i}(\bp_{k})}, \quad
%     A_{ij} =
%     \begin{cases}
%         1 & \text{if } i \in I_{j}, \\
%         0 & \text{if } i \notin I_{j}.
%     \end{cases}
% \end{equation}
% $A_{ij}$ indicates whether barrier $i$ bounds convex polytope $j$
% \begin{equation}
%     h(\bx) = \frac{1}{\kappa} \ln \bigg( \frac{1}{(\bPsi \cdot \bone) \bA} \cdot \bone \bigg) - \frac{b}{\kappa},
% \end{equation}
% \begin{equation}
%     \frac{\partial h}{\partial \bPsi}(\bx) = \frac{1}{{\rm e}^{\kappa h(\bx) + b}} \bigg( \bigg( \frac{1}{(\bPsi \cdot \bone) \bA}\bigg)^2 \bA^\top\bigg) * \bPsi,
% \end{equation}
% clean up formulas

\begin{example} \label{ex:ellipse}
Safe navigation of a polytope agent in a polytope environment is exemplified in Fig.~\ref{fig:ellipse}.
An ellipse agent is controlled to move around and ellipse obstacle, where both ellipses are approximated as polytopes with 32 vertices.
This leads to ${32 \times 32}$ barriers $\psi_{i}(\bp_{k})$ which are combined via~\eqref{eq:smoothing_body_general_polytope} with ${J = K = \{1, 2, \ldots, 32\}}$ and ${I_{j}=\{j\}}$.
The same controller and parameters are used as in Example~\ref{ex:L_shape}, except that now ${b=0}$.
The agent successfully navigates to the goal without collisions, as shown by the simulation results (with the same colors and notations as in Fig.~\ref{fig:L_shape}). 
\end{example}

\begin{figure}
\centering
\includegraphics[scale=1]{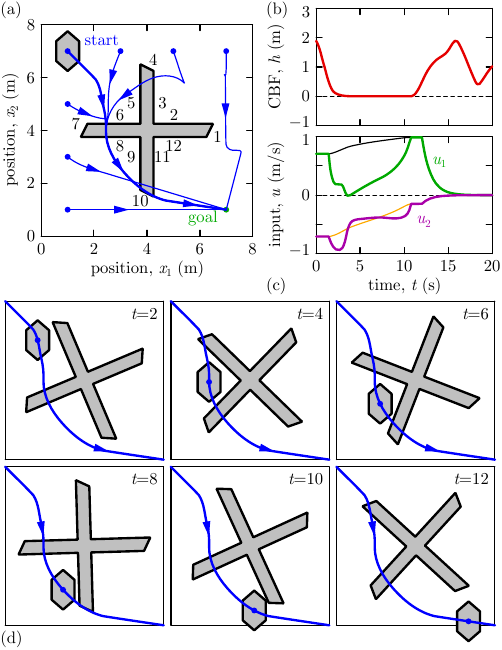}
\vspace{-3mm}
\caption{
Safety-critical navigation of a hexagon agent through a revolving door.
Safety is maintained using a time-varying counterpart of the proposed CBF candidate~\eqref{eq:smoothing_body_general_polytope}.
}
\vspace{-5mm}
\label{fig:revolving_door}
\end{figure}

\begin{example} \label{ex:revolving_door}
Another example of safe navigation is showcased in Fig.~\ref{fig:revolving_door} where a hexagonal agent is moving through a revolving door that has the shape of a non-convex dodecagon.
In this case, ${12 \times 6}$ barriers $\psi_{i}(\bp_{k},t)$ are combined via~\eqref{eq:smoothing_body_general_polytope} with ${K = \{1, 2, \ldots, 6\}}$.
The obstacle-free space is the union of 8 convex polytopes,
${J = \{1, 2, \ldots, 8\}}$,
bounded by barriers
${I_{1}=\{1\}}$, ${I_{2}=\{2,3\}}$,
${I_{3}=\{4\}}$, ${I_{4}=\{5,6\}}$,
${I_{5}=\{7\}}$, ${I_{6}=\{8,9\}}$,
${I_{7}=\{10\}}$, and ${I_{8}=\{11,12\}}$,
as numbered in the figure.
The barriers $\psi_{i}$ and CBF candidate $h$ are time-varying because the obstacle rotates (with 0.2 rad/s), and this is handled according to Remark~\ref{rem:timedependency}.
The controller and parameters match those in Example~\ref{ex:ellipse}.
The snapshots of the simulation results in Fig.~\ref{fig:revolving_door}(d) indicate that the agent manages to navigate through the door while its vertices slide along the obstacle, showing a minimally conservative behavior.

We remark that if the door was stationary the agent would safely stop at the door without reaching the goal.
This deadlock is due to using a single waypoint (the goal point) as motion plan, and it could be overcome via more waypoints.
This highlights that the proposed method is not meant to substitute motion planning but it helps to ensure safety online when following a not necessarily safe nominal motion plan.
\end{example}

\addtolength{\textheight}{-8mm}

\begin{example} \label{ex:pyramid}
Finally, we present an example of safe navigation in 3D space; see Fig.~\ref{fig:pyramid}.
A cube agent is controlled to move around a pyramid obstacle while staying above ground.
This yields ${6 \times 8}$ barriers $\psi_{i}(\bp_{k})$ that are combined via~\eqref{eq:smoothing_body_general_polytope} with ${K = \{1, 2, \ldots, 8\}}$.
The obstacle-free space is the union of 5 convex polytopes, given by
${J = \{1, 2, \ldots, 5\}}$,
${I_{1}=\{1\}}$, ${I_{2}=\{2,6\}}$,
${I_{3}=\{3,6\}}$, ${I_{4}=\{4,6\}}$, and
${I_{5}=\{5,6\}}$; see the indices of barriers (i.e., obstacle faces) in the figure.
The controller and parameters are the same as in Example~\ref{ex:ellipse}.
According to the simulation, the agent navigates without collision.
At the end of the motion, the proposed controller stops the agent above the goal (that is on the ground) to avoid collision between the bottom of the agent and the ground.
% Note that the goal point is located on the ground, so the bottom of the agent would collide with the ground if its center reached the goal.
% Therefore, the proposed controller stops the agent above the goal point to avoid collision.
\end{example}

\begin{figure}
\centering
\includegraphics[scale=1]{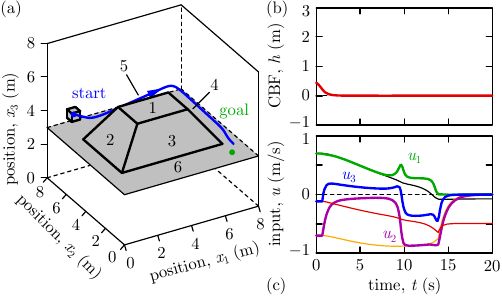}
\vspace{-3mm}
\caption{
Safety-critical navigation of a cube agent around a pyramid obstacle in 3D space.
Safety is maintained using the proposed CBF candidate~\eqref{eq:smoothing_body_general_polytope}.
}
\vspace{-5mm}
\label{fig:pyramid}
\end{figure}

Importantly, the CBF is given by a single closed form expression in each example, leading to a corresponding explicit control law without solving complex optimization problems.
This highlights the simplicity of our approach and enables real-time implementation (cf.~the simulation codes).

%%%%%%%%%%%%%%%%%%%%%%%%%%%%%%%%%%%%%%%%%%%%%%%%%%%%%%%%%%%%%%%%%%%%%%%%%%%%%%%%
\section{CONCLUSION}
\label{sec:concl}

In this paper, we introduced an algorithmic method to generate a control barrier function candidate in closed form for the safety-critical control of polytope-shaped agents navigating in polytope-shaped environments with formal guarantees of collision-free motion.
We established the method via an under-approximation and smoothing of the signed distance.
We showed examples of point and polytope agents navigating in not necessarily convex and potentially dynamically changing polytope environments, both in 2D and 3D.
While this work focused on the geometry of the problem for a single rigid-body agent governed by simplified kinematics, future research may address agents consisting of multiple bodies and extensions to more descriptive dynamical models.

\vspace{6pt}
\noindent \textbf{Acknowledgment.}  
We thank Calla Unruh and Avery Thomas for their numerical case studies in this topic.

% \addtolength{\textheight}{-6cm}

%%%%%%%%%%%%%%%%%%%%%%%%%%%%%%%%%%%%%%%%%%%%%%%%%%%%%%%%%%%%%%%%%%%%%%%%%%%%%%%%
\bibliographystyle{IEEEtran}
\bibliography{ccta_2024}	

\begin{thebibliography}{10}
\providecommand{\url}[1]{#1}
\csname url@samestyle\endcsname
\providecommand{\newblock}{\relax}
\providecommand{\bibinfo}[2]{#2}
\providecommand{\BIBentrySTDinterwordspacing}{\spaceskip=0pt\relax}
\providecommand{\BIBentryALTinterwordstretchfactor}{4}
\providecommand{\BIBentryALTinterwordspacing}{\spaceskip=\fontdimen2\font plus
\BIBentryALTinterwordstretchfactor\fontdimen3\font minus \fontdimen4\font\relax}
\providecommand{\BIBforeignlanguage}[2]{{%
\expandafter\ifx\csname l@#1\endcsname\relax
\typeout{** WARNING: IEEEtran.bst: No hyphenation pattern has been}%
\typeout{** loaded for the language `#1'. Using the pattern for}%
\typeout{** the default language instead.}%
\else
\language=\csname l@#1\endcsname
\fi
#2}}
\providecommand{\BIBdecl}{\relax}
\BIBdecl

\bibitem{AmesXuGriTab2017}
A.~D. Ames, X.~Xu, J.~W. Grizzle, and P.~Tabuada, ``Control barrier function based quadratic programs for safety critical systems,'' \emph{IEEE Transactions on Automatic Control}, vol.~62, no.~8, pp. 3861--3876, 2017.

\bibitem{garg2024advances}
K.~Garg, J.~Usevitch, J.~Breeden, M.~Black, D.~Agrawal, H.~Parwana, and D.~Panagou, ``Advances in the theory of control barrier functions: Addressing practical challenges in safe control synthesis for autonomous and robotic systems,'' \emph{Annual Reviews in Control}, vol.~57, p. 100945, 2024.

\bibitem{hamatani2020collision}
R.~Hamatani and H.~Nakamura, ``Collision avoidance control of robot arm considering the shape of the target system,'' in \emph{59th Annual Conference of the Society of Instrument and Control Engineers of Japan}, 2020, pp. 1311--1316.

\bibitem{almubarak2022safety}
H.~Almubarak, K.~Stachowicz, N.~Sadegh, and E.~A. Theodorou, ``Safety embedded differential dynamic programming using discrete barrier states,'' \emph{IEEE Robotics and Automation Letters}, vol.~7, no.~2, pp. 2755--2762, 2022.

\bibitem{mestres2024distributed}
P.~Mestres, C.~Nieto-Granda, and J.~Cortés, ``Distributed safe navigation of multi-agent systems using control barrier function-based controllers,'' \emph{IEEE Robotics and Automation Letters}, vol.~9, no.~7, pp. 6760--6767, 2024.

\bibitem{landi2019safetybarrier}
C.~T. Landi, F.~Ferraguti, S.~Costi, M.~Bonfè, and C.~Secchi, ``Safety barrier functions for human-robot interaction with industrial manipulators,'' in \emph{18th European Control Conference}, 2019, pp. 2565--2570.

\bibitem{singletary2022manipulation}
A.~W. Singletary, W.~Guffey, T.~Molnar, R.~Sinnet, and A.~D. Ames, ``Safety-critical manipulation for collision-free food preparation,'' \emph{IEEE Robotics and Automation Letters}, vol.~7, no.~4, pp. 10\,954--10\,961, 2022.

\bibitem{srinivasan2020synthesis}
M.~Srinivasan, A.~Dabholkar, S.~Coogan, and P.~A. Vela, ``Synthesis of control barrier functions using a supervised machine learning approach,'' in \emph{IEEE/RSJ International Conference on Intelligent Robots and Systems}, 2020, pp. 7139--7145.

\bibitem{long2021learning}
K.~Long, C.~Qian, J.~Cort{\'{e}}s, and N.~Atanasov, ``Learning barrier functions with memory for robust safe navigation,'' \emph{IEEE Robotics and Automation Letters}, vol.~6, no.~3, pp. 4931--4938, 2021.

\bibitem{vulcano2022}
V.~Vulcano, S.~G. Tarantos, P.~Ferrari, and G.~Oriolo, ``Safe robot navigation in a crowd combining {NMPC} and control barrier functions,'' in \emph{61st IEEE Conference on Decision and Control}, 2022, pp. 3321--3328.

\bibitem{abdi2023safecontrol}
H.~Abdi, G.~Raja, and R.~Ghabcheloo, ``Safe control using vision-based control barrier function ({V-CBF}),'' in \emph{IEEE International Conference on Robotics and Automation}, 2023, pp. 782--788.

\bibitem{keyumarsi2024lidar}
S.~Keyumarsi, M.~W.~S. Atman, and A.~Gusrialdi, ``{LiDAR}-based online control barrier function synthesis for safe navigation in unknown environments,'' \emph{IEEE Robotics and Automation Letters}, vol.~9, no.~2, pp. 1043--1050, 2024.

\bibitem{desa2024pointcloud}
M.~De~Sa, P.~Kotaru, and K.~Sreenath, ``Point cloud-based control barrier function regression for safe and efficient vision-based control,'' in \emph{IEEE International Conference on Robotics and Automation}, 2024, pp. 366--372.

\bibitem{dai2024sailing}
B.~Dai, R.~Khorrambakht, P.~Krishnamurthy, and F.~Khorrami, ``Sailing through point clouds: Safe navigation using point cloud based control barrier functions,'' \emph{IEEE Robotics and Automation Letters}, vol.~9, no.~9, pp. 7731--7738, 2024.

\bibitem{chen2024safer}
T.~Chen, A.~Swann, J.~Yu, O.~Shorinwa, R.~Murai, M.~{Kennedy III}, and M.~Schwager, ``{SAFER-Splat}: {A} control barrier function for safe navigation with online {Gaussian} splatting maps,'' \emph{arXiv preprint}, no. arXiv:2409.09868, 2024.

\bibitem{harms2024neural}
M.~Harms, M.~Kulkarni, N.~Khedekar, M.~Jacquet, and K.~Alexis, ``Neural control barrier functions for safe navigation,'' in \emph{IEEE/RSJ International Conference on Intelligent Robots and Systems}, 2024, pp. 10\,415--10\,422.

\bibitem{long2024safestabilizing}
K.~Long, K.~Tran, M.~Leok, and N.~Atanasov, ``Safe stabilizing control for polygonal robots in dynamic elliptical environments,'' in \emph{American Control Conference}, 2024, pp. 312--317.

\bibitem{thirugnanam2022safety}
A.~Thirugnanam, J.~Zeng, and K.~Sreenath, ``Safety-critical control and planning for obstacle avoidance between polytopes with control barrier functions,'' in \emph{IEEE International Conference on Robotics and Automation}, 2022, pp. 286--292.

\bibitem{thirugnanam2022duality}
------, ``Duality-based convex optimization for real-time obstacle avoidance between polytopes with control barrier functions,'' in \emph{American Control Conference}, 2022, pp. 2239--2246.

\bibitem{tayal2024polygonal}
M.~Tayal and S.~Kolathaya, ``Polygonal cone control barrier functions ({PolyC2BF}) for safe navigation in cluttered environments,'' in \emph{European Control Conference}, 2024, pp. 2212--2217.

\bibitem{cohen2023smooth}
M.~H. Cohen, P.~Ong, G.~Bahati, and A.~D. Ames, ``Characterizing smooth safety filters via the implicit function theorem,'' \emph{IEEE Control Systems Letters}, vol.~7, pp. 3890--3895, 2023.

\bibitem{glotfelter2017nonsmooth}
P.~Glotfelter, J.~Cort{\'{e}}s, and M.~Egerstedt, ``Nonsmooth barrier functions with applications to multi-robot systems,'' \emph{IEEE Control Systems Letters}, vol.~1, no.~2, pp. 310--315, 2017.

\bibitem{molnar2023composing}
T.~G. Molnar and A.~D. Ames, ``Composing control barrier functions for complex safety specifications,'' \emph{IEEE Control Systems Letters}, vol.~7, pp. 3615--3620, 2023.

\bibitem{lindemann2019stl}
L.~Lindemann and D.~V. Dimarogonas, ``Control barrier functions for signal temporal logic tasks,'' \emph{IEEE Control Systems Letters}, vol.~3, no.~1, pp. 96--101, 2019.

\bibitem{molnar2022modelfree}
T.~G. Molnar, R.~K.~Cosner, A.~W.~Singletary, W.~Ubellacker, and A.~D.~Ames, ``Model-free safety-critical control for robotic systems,'' \emph{IEEE Robotics and Automation Letters}, vol.~7, no.~2, pp. 944--951, 2022.

\bibitem{cohen2024reduced}
M.~H. Cohen, T.~G. Molnar, and A.~D. Ames, ``Safety-critical control for autonomous systems: Control barrier functions via reduced-order models,'' \emph{Annual Reviews in Control}, vol.~57, p. 100947, 2024.

\bibitem{taylor2022safebackstepping}
A.~J. Taylor, P.~Ong, T.~G. Molnar, and A.~D. Ames, ``Safe backstepping with control barrier functions,'' in \emph{61st IEEE Conference on Decision and Control}, 2022, pp. 5775--5782.

\bibitem{cohen2024constructive}
M.~H. Cohen, R.~K. Cosner, and A.~D. Ames, ``Constructive safety-critical control: Synthesizing control barrier functions for partially feedback linearizable systems,'' \emph{IEEE Control Systems Letters}, vol.~8, pp. 2229--2234, 2024.

\end{thebibliography}

\end{document}